\documentclass[11pt,a4paper]{article}
\pdfoutput=1 
\usepackage[utf8]{inputenc}

\usepackage[T1]{fontenc}
\usepackage{lmodern}

\usepackage[margin=1in]{geometry}
\usepackage{amsfonts}
\usepackage{amsmath}
\usepackage{amssymb}
\usepackage{amsthm}
\usepackage{float}
\usepackage[font=small,labelfont=bf]{caption} 
\allowdisplaybreaks
\usepackage{bbold}

\usepackage{physics}
\usepackage{xspace}

\usepackage{multirow}

\usepackage{bold-extra} 

\usepackage{hyperref}
\usepackage[svgnames]{xcolor}
\hypersetup{colorlinks={true},urlcolor={blue},linkcolor={DarkBlue},citecolor=[named]{DarkGreen}}

\usepackage{microtype}
\usepackage[capitalise,nameinlink,noabbrev]{cleveref}

\usepackage{tikz}
\usetikzlibrary{arrows.meta}

\usepackage{doi}

\usepackage{graphicx}

\def\dosth#1{\ifx###1##\else\dofirst#1\anytoken\fi}
\def\doagain#1\anytoken{\dosth{#1}}
\def\payoffpairs#1#2#3{\m=#1\multiply\m by 4 \advance\m by -1 \n=1
  \def\dofirst##1{\put(\n,-\m){\makebox(0,0){\strut##1}}\advance\n by 4 \doagain}%
  \dosth{#2\strut}%
  \m=#1\multiply\m by 4 \advance\m by -3 \n=3 \dosth{#3\strut}}
\def\singlepayoffs#1#2{\m=#1\multiply\m by 4 \advance\m by -2 \n=2
  \def\dofirst##1{\put(\n,-\m){\makebox(0,0){\strut##1}}\advance\n by 4 \doagain}%
  {\large\dosth{#2\strut}}}
\newcommand{\bimatrixgame}[8]{%
\setlength{\unitlength}{#1}%
\newcount\rows
\newcount\cols
\rows=#2
\cols=#3
\newcount\rowcoord
\newcount\colcoord
\rowcoord=\rows
\colcoord=\cols
\multiply\rowcoord by 4
\multiply\colcoord by 4
\newcount\m
\newcount\n
\m=\rowcoord
\n=\colcoord
\advance\m by 2 
\advance\n by 2 
\begin{picture}(\n,\m)(-2,-\rowcoord)
\m=\rows
\n=\cols
\advance\m by 1
\advance\n by 1 
\thinlines
\multiput(0,0)(0,-4){\m}{\line(1,0){\colcoord}}
\multiput(0,0)(4,0){\n}{\line(0,-1){\rowcoord}}
\put(0,0){\line(-1,1){2}}
\put(-1.5,0.5){\makebox(0,0)[r]{#4}}  
\put(-.7,1.7){\makebox(0,0)[l]{#5}}   
\n=2
\def\dofirst##1{\put(-0.8,-\n){\makebox(0,0)[r]{\strut##1}}\advance\n by 4
   \doagain}
\dosth{#6\strut} 
\n=2
\def\dofirst##1{\put(\n,1.0){\makebox(0,0){\strut##1}}\advance\n by 4
   \doagain}
\dosth{#7\strut}#8%
\end{picture}}
\usepackage{enumerate}

\theoremstyle{definition}

\theoremstyle{plain}
\newtheorem{theorem}{Theorem}[section]
\newtheorem{lemma}[theorem]{Lemma}
\newtheorem{corollary}[theorem]{Corollary}

\theoremstyle{remark}

\Crefname{claim}{Claim}{Claims}

\def\fp/{\textup{\textsf{FP}}}
\def\p/{\textup{\textsf{P}}}
\def\np/{\textup{\textsf{NP}}}
\def\conp/{\textup{\textsf{co-NP}}}
\def\fnp/{\textup{\textsf{FNP}}}
\def\tfnp/{\textup{\textsf{TFNP}}}
\def\ptfnp/{\textup{\textsf{PTFNP}}}
\def\ppa/{\textup{\textsf{PPA}}}
\def\ppad/{\textup{\textsf{PPAD}}}
\def\ppads/{\textup{\textsf{PPADS}}}
\def\ppp/{\textup{\textsf{PPP}}}
\def\pwpp/{\textup{\textsf{PWPP}}}
\def\pls/{\textup{\textsf{PLS}}}
\def\cls/{\textup{\textsf{CLS}}}
\def\ppadpls/{\textup{$\textsf{PPAD} \cap \textsf{PLS}$}}
\def\ppapls/{\textup{$\textsf{PPA} \cap \textsf{PLS}$}}
\def\eopl/{\textup{\textsf{EOPL}}}
\def\sopl/{\textup{\textsf{SOPL}}}
\def\ueopl/{\textup{\textsf{UEOPL}}}
\def\fixp/{\textup{\textsf{FIXP}}}
\def\bu/{\textup{\textsf{BU}}}
\def\bbu/{\textup{\textsf{BBU}}}
\def\linearfixp/{\textup{\textsf{Linear-FIXP}}}
\def\pspace/{\textup{\textsf{PSPACE}}}

\def\pcircuit{\textup{\textsc{Pure-Circuit}}\xspace}

\newcommand{\garbo}{\ensuremath{\bot}\xspace}

\newcommand{\innei}[1]{\ensuremath{N^{-}(#1)}\xspace}
\newcommand{\outnei}[1]{\ensuremath{N^{+}(#1)}\xspace}
\newcommand{\nei}[1]{\ensuremath{N(#1)}\xspace}

\newcommand{\val}[1]{\boldsymbol{\mathrm{x}}[#1]}
\newcommand{\valonly}{\boldsymbol{\mathrm{x}}}

\newcommand{\PURE}{\textup{\textsf{PURIFY}}\xspace}

\newcommand{\NOT}{\textup{\textsf{NOT}}\xspace}

\newcommand{\AND}{\textup{\textsf{AND}}\xspace}

\newcommand{\supp}{\textup{\textsf{supp}}\xspace}

\newcommand{\wn}[1]{\ensuremath{#1}-WSNE\xspace}

\newcommand{\zero}{\textup{\textsf{zero}}\xspace}
\newcommand{\one}{\textup{\textsf{one}}\xspace}

\newcommand{\eps}{\ensuremath{\varepsilon}\xspace}

\newcommand{\vba}{\ensuremath{\vb{a}}\xspace}
\newcommand{\vbai}{\ensuremath{\vb{a}_{-i}}\xspace}
\newcommand{\vbs}{\ensuremath{\vb{s}}\xspace}
\newcommand{\vbsi}{\ensuremath{\vb{s}_{-i}}\xspace}

\newcommand{\br}{\ensuremath{\text{br}}\xspace}

\title{Tight Inapproximability for Graphical Games}

\author{
\begin{tabular}{cc}
& \\
\textbf{Argyrios Deligkas} & \textbf{John Fearnley}\\
\small{Royal Holloway, United Kingdom} & \small{University of Liverpool, United Kingdom}\\
\href{mailto:argyrios.deligkas@rhul.ac.uk}{\small{\texttt{argyrios.deligkas@rhul.ac.uk}}} & \href{mailto:john.fearnley@liverpool.ac.uk}{\small{\texttt{john.fearnley@liverpool.ac.uk}}}\\
& \\
\textbf{Alexandros Hollender} & \textbf{Themistoklis Melissourgos}\\
\small{University of Oxford, United Kingdom} & \small{University of Essex, United Kingdom}\\
\href{mailto:alexandros.hollender@cs.ox.ac.uk}{\small{\texttt{alexandros.hollender@cs.ox.ac.uk}}} & \href{mailto:themistoklis.melissourgos@essex.ac.uk}{\small{\texttt{themistoklis.melissourgos@essex.ac.uk}}}\\
& \\
\end{tabular}
}

\date{}

\begin{document}

\maketitle
\thispagestyle{empty}

\begin{abstract}
	We provide a complete characterization for the computational complexity of finding approximate equilibria in two-action graphical games. We consider the two most well-studied approximation notions: $\eps$-Nash equilibria ($\eps$-NE) and $\eps$-well-supported Nash equilibria ($\eps$-WSNE), where $\eps \in [0,1]$. We prove that computing an $\eps$-NE is \ppad/-complete for any constant $\eps < 1/2$, while a very simple algorithm (namely, letting all players mix uniformly between their two actions) yields a $1/2$-NE. On the other hand, we show that computing an $\eps$-WSNE is \ppad/-complete for any constant $\eps < 1$, while a $1$-WSNE is trivial to achieve, because any strategy profile is a $1$-WSNE. All of our lower bounds immediately also apply to graphical games with more than two actions per player.
\end{abstract}

\section{Introduction}

Graphical games were introduced more than twenty years ago
by Kearns, Littman, and Singh~\cite{KearnsLS01-graphical-games} as a succinct model of a multi-player game.
These games have found a wide variety of applications. On the theoretical
side, they have served as a fundamental tool for showing seminal \ppad/-completeness
results in algorithmic game theory~\cite{DaskalakisGP09-Nash,ChenDT09-Nash}. Practically,
graphical games have been used as a foundation for the game theoretic analysis of
networks~\cite{galeotti2010network,jackson2015games}, social networks, and
multi-agent systems~\cite{Kearns07-AGT-graphical-games,jackson2011overview}. 

A graphical game is specified by a directed graph with $n$ vertices. Each vertex 
represents a player, and each player has $m$ distinct actions.
The edges of the graph specify the interactions between the players: the payoff
to player $i$ is determined entirely by the actions chosen by player $i$ and
the in-neighbours of player $i$.
Formally, the payoffs for a player are given
by a payoff tensor, which maps the actions chosen by that player and their
in-neighbours to a payoff in $[0, 1]$. 

Graphical games are more succinct than standard normal form games when the
maximum in-degree $d$ is constant.
An $n$-player $m$-action
game requires $n \cdot m^n$ payoffs to be written down, which
becomes infeasibly large as $n$ grows. On the other hand, 
each tensor in a graphical game has $m^{d+1}$ payoff entries, giving $n
\cdot m^{d+1}$ payoffs in total, which provides much more reasonable scaling as
$n$ grows when $d$ is constant.

\paragraph{\bf The complexity of finding equilibria.}

Unfortunately it is known that finding a Nash equilibrium in a graphical
game is a \ppad/-hard problem~\cite{DaskalakisGP09-Nash} and thus considered to be intractable. 
This has left open
the question of finding \emph{approximate} Nash equilibria, and two notions of
approximate equilibrium have been studied in the literature. An
\emph{$\eps$-Nash equilibrium} (\eps-NE) requires that no player can improve
their payoff by more than $\eps$ by unilaterally changing their strategy, while
an \emph{$\eps$-well-supported Nash equilibrium} ($\eps$-WSNE) requires that
all players only place positive probability on actions that are $\eps$-best
responses.
Every \eps-WSNE is also an \eps-NE, but the reverse is not true. 

In this paper we make the standard assumption that all payoffs lie in the range
$[0, 1]$, which then gives us a scale on which we can measure the additive approximation factor $\eps$. A $0$-NE
or $0$-WSNE is an exact Nash equilibrium, while a $1$-NE or $1$-WSNE can be
trivially obtained, since the requirements will be satisfied no matter what
strategies the players use.

For many years, the best known lower bounds for approximate equilibria in
graphical games were given by Rubinstein~\cite{Rubinstein18-Nash-inapproximability}, who proved that there is
some unspecified small constant $\eps$ for which finding an $\eps$-NE is
\ppad/-complete, and there is a different but still unknown small
constant $\eps'$ for which finding an $\eps'$-WSNE is \ppad/-complete.
In fact, Rubinstein's result applies to games that are simultaneously graphical games of constant degree and also \emph{polymatrix games}, namely in which each edge represents a two-player game and a player's payoff is the sum of payoffs from these games against her in-neighbours.

This was recently improved by a result of Deligkas et al.~\cite{DFHM22}. They
showed that it is \ppad/-complete to find a $1/48$-NE of a two-action
polymatrix game, and it is \ppad/-complete to find an $\eps$-WSNE of a
two-action polymatrix game for all $\eps < 1/3$ (the latter result being tight). Since these hardness results hold even for constant-degree polymatrix games, they also apply to graphical
games.\footnote{A polymatrix game of constant degree can be turned into its graphical game representation in polynomial time.}

On the other hand, only trivial upper bounds are known for approximate equilibria in graphical games, even when the players only have two actions.
For \eps-WSNE the upper bound is $1$, since any strategy profile is a 1-WSNE.
For \eps-NE, the upper bound is $1/2$ in two-action games and is achieved when all players uniformly mix over their two actions; the upper bound simply follows from the fact that every player plays their best response with probability at least $0.5$ and that the maximum payoff is bounded by 1.

\paragraph{\bf Our Contribution.}
In this paper we show that the aforementioned trivial upper bounds are in fact the {\em best possible}, by providing matching lower bounds for finding approximate equilibria in graphical games.

For the problem of finding an $\eps$-WSNE, we show that it is \ppad/-complete to
find an $\eps$-WSNE in a two-action graphical game for every constant $\eps < 1$. Since finding a $1$-WSNE is trivial, we
obtain a striking characterization for constant $\eps$: no polynomial-time
algorithm can find a non-trivial WSNE of a graphical game unless $\ppad/ = \p/$.

In fact, we present a more fine-grained analysis that provides a complete dichotomy of the complexity of finding a WSNE in a two-action graphical game of maximum in-degree $d$:
\begin{itemize}
	\item  a $\left(1 - \frac{2}{2^d + 1}\right)$-WSNE can be found in polynomial time;
	\item  for any $\eps < 1 - \frac{2}{2^d + 1}$, it is \ppad/-complete to compute a 
	$\eps$-WSNE.
\end{itemize}
Thus, for any constant $\eps < 1$ there exists a sufficiently large constant in-degree $d$, such that the problem becomes intractable.

For $\eps$-NE we show that it is \ppad/-complete to find an $\eps$-NE of a
two-action graphical game for any constant $\eps < 0.5$; this complements the
straightforward algorithm for finding a $0.5$-NE in a two-action graphical
game. We note that our lower bounds, both for $\eps$-WSNE and $\eps$-NE, also hold for graphical games with more than two actions.\footnote{We can simply add additional ``dummy'' actions that are just copies of the original two actions.}

All of our lower bounds are shown via reductions from the \pcircuit problem
that was recently introduced 
by Deligkas et al.~\cite{DFHM22}. In that paper, \pcircuit was used to show the aforementioned lower
bounds for polymatrix games. We
show that \pcircuit can likewise be used to show stronger and tight lower bounds for
graphical games.

\paragraph{\bf Further Related Work.}
The class \ppad/ was defined by Papadimitriou~\cite{Papadimitriou94-TFNP-subclasses}. 
Many years later, Daskalakis, Goldberg, and Papadimitriou~\cite{DaskalakisGP09-Nash} proved that finding an \eps-NE in graphical games and 3-player normal form games is \ppad/-complete for an exponentially small \eps. These results were further extended to polynomially small \eps for 2-player games and two-action polymatrix games with bipartite underlying graph by Chen, Deng, and Teng~\cite{ChenDT09-Nash}. 

On the positive side, Elkind, Goldberg, and Goldberg~\cite{EGG} derived a polynomial-time algorithm on two-action graphical games on paths and cycles, and Ortiz and Irfan~\cite{OrtizIrfan-hypergraphical} derived an approximation scheme for constant-action graphical games on trees. 
For polymatrix games, Barman, Ligett, and Piliouras~\cite{BLP-polymatrix-trees} derived a quasi-PTAS on trees, and Deligkas, Fearnley, and Savani~\cite{DFS-treewidth} derived a quasi-PTAS for constant-action games on bounded treewidth graphs. 
These results where complemented by the same authors~\cite{DFS20-ptph} who showed that finding an {\em exact} NE is \ppad/-complete for polymatrix games on trees when every player has 20 actions.

\section{Preliminaries}
For every natural number $k$, let $\Delta^k$ denote the $k$-dimensional simplex, i.e., $\Delta^k := \{x \in \mathbb{R}^{k+1} \; : \; x \ge 0,\; \sum_{i=1}^{k+1} x_i = 1\}$, and let $[k] := \{1, 2, \ldots, k\}$.

\subsection{Graphical Games}

An $n$-player $m$-action graphical game is defined by a directed graph $G = (V,E)$, where
$|V|=n$, each node of $G$ corresponds to a player, and the maximum number of actions per player is $m$.
We define the in-neigbourhood of node $i$ to be $\innei{i} := \{ j \in V  :  (j,i) \in E \}$, and similarly its out-neighbourhood to be $\outnei{i} := \{ j \in V : (i,j) \in E \}$. We also define the neighbourhood of $i$ as $\nei{i} := \innei{i} \cup \outnei{i} \cup \{i\}$. We include $i$ in its own neighbourhood for notational convenience. 

In a graphical game, each player $i$ participates in a normal form game $G_i$ whose player set is $\nei{i}$, but she affects only the payoffs of her out-neighbours. Player $i$ has $m_i$ \emph{actions}, or \emph{pure strategies}, and her payoffs are represented by a function $R_{i} : [m_i] \times \prod_{j \in \innei{i}} [m_j] \mapsto [0,1]$ which will be referred to as the \emph{payoff tensor} of $i$. If the codomain of $R_i$ is $\{0,1\}$ for all $i \in V$, we have a {\em win-lose} graphical game.

A \emph{mixed strategy} $s_i$ for player $i$ specifies a probability distribution over
player $i$'s actions. Thus, the set of mixed strategies for player $i$ corresponds to the $(m_i-1)$-dimensional simplex $\Delta^{m_i-1}$. 
The \emph{support}
of a mixed strategy $s_i = (s_{i}(1), s_{i}(2), \dots, s_{i}(m_i)) \in \Delta^{m_i-1}$ is given by $\supp(s_i) = \{j \in [m_i] \; : \; s_i(j) > 0\}$. In other words, it is the set of pure strategies that are played with non-zero probability in strategy $s_i$.

An action profile $\vba(H) := (a_i)_{i \in H}$ over a player set $H$ is a tuple of actions, one for each player in $H$, and so the set of these action profiles is given by $A(H) = \prod_{i \in H} [m_i]$.
Similarly, a \emph{strategy profile} $\vbs(H)$ over the same set is a tuple of mixed strategies, and so the set of these strategy profiles is given by $\prod_{i \in H} \Delta^{m_i-1}$.
We define the \emph{partial action profile} \vbai to be the tuple of all players' actions except $i$'s action, and similarly we define the \emph{partial strategy profile} \vbsi. 
The expected payoff of $i$ when she plays action $k \in [m_i]$, and all
other players play according to $\vbsi$ is 
\begin{equation*}
	u_i(k, \vbsi) := \sum_{\vba \in A(\innei{i})} R_{i}(k; \vba) \cdot \prod_{{j \in \innei{i}}} s_j(a_j).
\end{equation*}
Notice that this depends only on the in-neighbours of $i$.

The expected payoff to player $i$ under $\vbs$ is therefore
\begin{equation*}
	u_i(\vbs) := \sum_{k \in [m_i]} u_i(k, \vbsi) \cdot s_i(k).
\end{equation*}

A pure strategy $k$ is a \emph{best response} for player $i$ against a partial strategy profile \vbsi
if it achieves the maximum payoff over all her pure strategies, that is,
\begin{equation*}
	u_i(k, \vbsi) = \max_{\ell \in [m_i]} u_i(\ell, \vbsi).
\end{equation*}
Pure strategy $k$ is an \emph{$\eps$-best response} if the payoff it yields is within $\eps$ of a
best response, meaning that 
\begin{equation*}
	u_i(k, \vbsi) \ge \max_{\ell \in [m_i]} u_i(\ell, \vbsi) - \eps.
\end{equation*}
Finally, the \emph{best response payoff} for player $i$ is 
\begin{equation*}
	\br_i(\vbs_{-i}) := \max_{\ell \in [m_i]} u_i(\ell, \vbsi).
\end{equation*}

\begin{figure*}[t!]
	\centering
	\begin{minipage}{0.27\textwidth}
		\begin{center}
			\begin{tabular}{c||c}
				$u$ & $v$ \\ \hline
				0 & 1 \\
				1 & 0 \\
				$\garbo$ & $\{0, 1, \garbo\}$\\
				\multicolumn{2}{c}{}
			\end{tabular}
			\caption*{\NOT gate}
		\end{center}
	\end{minipage}
	\begin{minipage}{0.35\textwidth}
		\begin{center}
			\begin{tabular}{c|c||c}
				$u$ & $v$ & $w$ \\ \hline
				1 & 1 & 1 \\
				0 & $\{0, 1, \garbo\}$ & 0 \\
				$\{0, 1, \garbo\}$ & 0 & 0 \\
				\multicolumn{2}{c||}{Else} & $\{0, 1, \garbo\}$
			\end{tabular}
			\caption*{\AND gate}
		\end{center}
	\end{minipage}
	\begin{minipage}{0.35\textwidth}
		\begin{center}
			\begin{tabular}{c||c|c}
				$u$ & \phantom{xx}$v$\phantom{xx}  & $w$ \\ \hline
				$0$ & $0$ & $0$ \\
				$1$ & $1$ & $1$ \\
				\multirow{2}{*}{$\garbo$} & \multicolumn{2}{c}{At least one} \\
				& \multicolumn{2}{c}{output in  $\{0, 1\}$}
			\end{tabular}
			\caption*{\PURE gate}
		\end{center}
	\end{minipage}
	\caption{The truth tables of the three gates of \pcircuit.}
	\label{fig:gates}
\end{figure*}

\paragraph{\bf (Approximate) Nash equilibria.}
A strategy profile \vbs is a \emph{Nash equilibrium} if $\br_i(\vbs) = u_i(\vbs)$ for all players $i$, i.e., every player achieves their best response payoff.
A strategy profile \vbs is an \emph{$\eps$-Nash equilibrium} (\eps-NE) if every
player's payoff is within $\eps$ of their best response payoff, meaning that
$u_i(\vbs) \ge \br_i(\vbs) - \eps$. 
A strategy profile \vbs is an \emph{$\eps$-well supported Nash equilibrium}
(\eps-WSNE) if every
player only plays strategies that are $\eps$-best responses, meaning that
for all $i$ we have that $\supp(s_i)$ contains only 
$\eps$-best response strategies.

\subsection{The \pcircuit Problem}

An instance of the \pcircuit problem is given by a node set $V=[n]$ and a set $G$ of gate-constraints (or just \emph{gates}). Each gate $g \in G$ is of the form $g = (T,u,v,w)$ where $u,v,w \in V$ are distinct nodes, and $T \in \{\NOT, \AND, \PURE\}$ is the type of the gate, with the following interpretation.
\begin{itemize}
	\item If $T=\NOT$, then $u$ is the input of the gate, and $v$ is its output. ($w$ is unused)
	\item If $T=\AND$, then $u$ and $v$ are the inputs of the gate, and $w$ is its output.
	\item If $T=\PURE$, then $u$ is the input of the gate, and $v$ and $w$ are its outputs.
\end{itemize}
We require that each node is the output of exactly one gate.

A solution to instance $(V,G)$ is an assignment $\valonly: V \to \{0,1,\garbo\}$ that satisfies all the gates (see Fig.~\ref{fig:gates}), i.e., for each gate $g=(T,u,v,w) \in G$ we have the following. 
\begin{itemize}
	\item If $T=\NOT$ in $g=(T,u,v)$, then $\valonly$ satisfies
	\begin{align*}
		&\val{u} = 0 \implies \val{v} = 1\\
		&\val{u} = 1 \implies \val{v} = 0.
	\end{align*}
	\item If $T=\AND$ in $g = (T,u,v,w)$, then $\valonly$ satisfies
	\begin{align*}
		\val{u} = \val{v} = 1 & \implies \val{w} = 1\\
		x[u]=0 \lor x[v]=0 &\implies x[w]=0.
	\end{align*}
	\item If $T=\PURE$, then $\valonly$ satisfies
	\begin{align*}
		& \{\val{v}, \val{w}\} \cap \{0,1\} \neq \emptyset\\
		& \val{u} \in \{0,1\} \implies \val{v} = \val{w} = \val{u}.
	\end{align*}
\end{itemize}
The structure of a \pcircuit instance is captured by its \emph{interaction graph}. This graph is constructed on the vertex set $V = [n]$ by adding a directed edge from node $u$ to node $v$ whenever $v$ is the output of a gate with input $u$. The total degree of a node is the sum of its in- and out-degrees.

\begin{theorem}[\cite{DFHM22}]
	\label{thm:pancircuit}
	\pcircuit is \ppad/-complete, even when every node of the interaction graph has in-degree at most 2 and total degree at most 3.
\end{theorem}

\section{Well-Supported Nash Equilibria}

\paragraph{An easy upper bound.}

We present a polynomial time algorithm to compute a \wn{\left(1 - \frac{2}{2^d + 1}\right)} in any two-action graphical game with maximum in-degree $d \geq 2$. The algorithm relies on a simple and natural approach that has been used for similar problems by Liu et al.~\cite{LiuLD21-sparse-winlose-polymatrix} and Deligkas et al.~\cite{DFHM22}.

The algorithm proceeds in two steps. In the first step, it iteratively checks for a player that has an action that is {\em \eps-dominant}; an action which, if played with probability 1, will satisfy the \wn{\eps} conditions no matter what strategies the in-neighbours play.
Here, we will set $\eps =  1 - \frac{2}{2^d + 1}$, where $d$ is the maximum in-degree of the graph.
If such a player with an \eps-dominant action exists, the algorithm fixes the strategy of that player, updates the game accordingly, and iterates until there is no such player left.
In the second step, the algorithm lets all remaining players mix uniformly, i.e., every player without an \eps-dominant action plays each of its two actions with probability $1/2$.

\begin{theorem}
	\label{thm:graphical-wsne-UB}
	There is a polynomial-time algorithm that finds a \wn{\left(1 - \frac{2}{2^d + 1}\right)} in a two-action graphical game with maximum in-degree $d$.
\end{theorem}
\begin{proof}
	It is easy to see that the algorithm described above runs in polynomial time. In particular, we can check if a player has an \eps-dominant action by simply going over all possible action profiles of its in-neighbours. We will prove it computes an \wn{\eps} for $\eps = 1 - \frac{2}{2^d + 1}$.
	By definition of \eps-dominance, the players whose actions were fixed in the first step satisfy the constraints of \eps-WSNE.
	Notice that after fixing any such player to play an \eps-dominant action, we get a smaller graphical game. Thus, it suffices to consider the graphical game we obtain after the end of the first step, and to show that if all (remaining) players mix uniformly, this is an \wn{\eps}.
	
	So, consider a player $i$ in the graphical game we obtain after the end of the first step, and denote its in-degree by $k := |\innei{i}| \leq d$. Since all in-neighbours of player $i$ are mixing uniformly, the expected payoff of player $i$ for playing action $0$ is $\sum_{\vba \in A(\innei{i})} R_{i}(0 ; \vba)/2^k$, and for playing action $1$, it is $\sum_{\vba \in A(\innei{i})} R_{i}(1 ; \vba)/2^k$. The constraints of \wn{\eps} for player $i$ are satisfied if both actions are $\eps$-best responses, i.e., if
	\begin{align*}
	    \left|\sum_{\vba \in A(\innei{i})} R_{i}(0 ; \vba)/2^k - \sum_{\vba \in A(\innei{i})} R_{i}(1 ; \vba)/2^k\right| \leq \eps.
	\end{align*}
	This can be rewritten as $ \frac{1}{2^k} \left| \sum_{\vba \in A(\innei{i})} f_{i}(\vba) \right| \leq \eps $, where, for any action profile $\vba \in A(\innei{i})$, we let
	\begin{align*}
		f_{i}(\vba) := R_{i}(0 ; \vba) - R_{i}(1 ; \vba).
	\end{align*}
	Note that since all payoffs lie in $[0,1]$, we always have $f_{i}(\vba) \in [-1,1]$.
	
	Let $M := \sum_{\vba \in A(\innei{i})} f_{i}(\vba)$. To prove the correctness of the algorithm, it suffices to prove that $\left| M \right| \leq 2^k - 1 - \eps$; since then 
	$ \frac{1}{2^k} \cdot |M| \leq 1- \frac{1+\eps}{2^k} \leq 1- \frac{1+\eps}{2^d} = \eps$. To see why this is indeed the case, observe the following. Since player $i$ does not have an \eps-dominant action (otherwise, it would have been removed in the first step), it means that there exist $\vba, \vba' \in A(\innei{i})$  such that 
	\begin{align}
		\label{eq:graphical-ub}
		f_{i}(\vba) > \eps \quad \text{and} \quad  f_{i}(\vba') < -\eps.
	\end{align}
	Given that $M$ is the sum of $2^k$ terms, each of them upper bounded by $1$, and at least one of them upper bounded by $-\eps$ by \eqref{eq:graphical-ub}, it follows that $M \leq 2^k-1 -\eps$. Similarly, since each term is also lower bounded by $-1$, and one of them is lower bounded by $\eps$, we also obtain that $M \geq -2^k-1 + \eps$. Thus, $\left| M \right| \leq 2^k - 1 - \eps$, as desired, and this completes the proof of correctness.
\end{proof}

\paragraph{\bf The lower bound.}
In this section we prove a matching lower bound for Theorem~\ref{thm:graphical-wsne-UB},
which essentially proves that computing an $\eps$-WSNE in two-action graphical games is \ppad/-complete for every constant $\eps \in (0,1)$.

We will prove our result by a reduction from \pcircuit. For the remainder of this section, we fix $\eps < 1 - \frac{2}{2^d+1}$.
Given a \pcircuit instance with in-degree 2, we build a two-action graphical game, where the two actions will be named $\zero$ and $\one$. For any given $d \geq 2$, the game will have in-degree at most $d$.
Each node $v$ of the \pcircuit instance will correspond to a player in the 
game -- the game will have some additional auxiliary players too -- whose strategy in any \eps-WSNE will encode a solution to the \pcircuit problem as follows. 
Given a strategy $s_v$ for the player that corresponds to node $v$, we define the assignment $\valonly$ for \pcircuit such that:
\begin{itemize}
	\item if $s_{v}(\zero) = 1$, then $\val{v} = 0$;
	\item if $s_{v}(\one) = 1$, then $\val{v} = 1$;
	\item otherwise, $\val{v} = \garbo$.
\end{itemize}

We now give implementations for \NOT, \AND, and \PURE gates. We note that, in
all three cases, the payoff received by player $v$ is only affected by the
actions chosen by the players representing the inputs to the (unique) gate $g$
that outputs to $v$. Thus, we can argue about the equilibrium condition at $v$
by only considering the players involved in gate $g$, and we can ignore all
other gates while doing this.

\paragraph{\NOT gates.}
For a gate $g = (\NOT, u, v)$ -- where recall that $u$ is the input variable and $v$ is the output variable -- we create a gadget involving players $u$ and $v$, where player $v$ has a unique incoming edge from $u$. The payoffs of $v$ are defined as follows.
\begin{itemize}
	\item If $u$ plays $\zero$, then $v$ gets payoff 0 from playing $\zero$ and payoff
	1 from playing  $\one$.
	\item If $u$ plays $\one$, then $v$ gets 1 from $\zero$ and 0 from $\one$.
\end{itemize}
This gadget appeared in~\cite{DFHM22}, but we include it here for completeness.
We claim that this gadget works correctly.
\begin{itemize}
	\item[-] If $s_{u}(\zero) = 1$, i.e., $u$ encodes 0, observe that for player $v$ action $\zero$ yields payoff 0, while action $\one$ yields payoff 1. Hence, by the constraints imposed by \eps-WSNE it must hold that	$s_v(\one)=1$, and thus $v$ encodes $1$.
	\item[-] Using identical reasoning, we can prove that if $s_{u}(\one) = 1$, then $s_{v}(\one) = 0$ in any $\eps$-WSNE.
\end{itemize}

\paragraph{\bf \AND gates.}
For a gate $g = (\AND, u, v, w)$ we create the following gadget with players $u, v$ and $w$, where $u$ and $v$ are the in-neighbors of $w$. The payoffs of $w$ are as follows.
\begin{itemize}
	\item If $s_u(\one) = 1$ and $s_v(\one) = 1$, then $w$ gets payoff 0 from playing $\zero$ and payoff 1 from playing $\one$.
	\item For any other action profile of $u$ and $v$, player $w$ gets 1 from $\zero$ and 0 from $\one$.
\end{itemize}
Next we argue that this gadget works correctly.
\begin{itemize}
	\item[-] If $s_u(\one) = 1$ and $s_v(\one) = 1$, i.e. both $u$ and $v$ encode 1, observe that for player $w$ action \zero yields payoff 0, while action \one yields payoff 1. Hence, by the constraints imposed by \eps-WSNE it must hold that	$s_w(\one)=1$, and thus $w$ encodes $1$. 
	\item[-] If at least one of $u$ or $v$ encodes 0, then for player $w$ action \zero yields expected payoff 1 while action \one yields expected payoff 0. Hence, by the constraints imposed by \eps-WSNE it must hold that	$s_w(\zero)=1$, and thus $w$ encodes $0$.
\end{itemize}

\paragraph{\bf \PURE gates.}
For a gate $g = (\PURE, u, v, w)$ we create the following gadget with $d+3$ players.
We introduce auxiliary players $u_1, u_2, \ldots, u_d$. Each player $u_i$ has a unique incoming edge from $u$. The idea is that in any \eps-WSNE, every player $u_i$ ``copies'' the strategy of player $u$.
\begin{itemize}
	\item If $u$ plays \zero, then $u_i$ gets 1 from \zero and 0 from \one.
	\item If $u$ plays \one, then $u_i$ gets 0 from zero and 1 from \one.
\end{itemize}

\begin{lemma}\label{lem:graphical-wsne-aux}
	At any \eps-WSNE the following hold for every $i \in [d]$: if $s_u(\zero)=1$, then $s_{u_i}(\zero)=1$;
	if $s_u(\one)=1$, then $s_{u_i}(\one)=1$.
\end{lemma}
\begin{proof}
	If $s_u(\zero)=1$, then for player $u_i$ action \zero yields payoff 1, while action \one yields payoff 0. Thus, the constraints of \eps-WSNE dictate that $s_{u_i}(\zero)=1$.
	If $s_u(\one)=1$, then for player $u_i$ action \zero yields payoff 0, while action \one yields payoff 1. Thus, the constraints of \eps-WSNE dictate that $s_{u_i}(\one)=1$.
\end{proof}
Next, we describe the payoff tensors of players $v$ and $w$; each one of them has in-degree $d$ with edges from all $u_1, u_2, \ldots, u_d$. In what follows, fix $\lambda := 1- \frac{2}{2^d+1}$. The payoffs of player $v$ are as follows.
\begin{itemize}
	\item If $v$ plays \zero and at least one of $u_1,\ldots, u_k$ plays \zero, then the payoff for $v$ is 1.
	\item If $v$ plays \zero and every one of $u_1,\ldots, u_k$ plays \one, then the payoff for $v$ is 0.
	\item If $v$ plays \one and at least one of $u_1,\ldots, u_k$ plays \zero, then the payoff for $v$ is 0.
	\item If $v$ plays \one and every one of $u_1,\ldots, u_k$ plays \one, then the payoff for $v$ is $\lambda$.
\end{itemize}
The payoffs of player $w$ are as follows.
\begin{itemize}
	\item If $w$ plays \zero and every one of $u_1,\ldots, u_k$ plays \zero, then the payoff for $w$ is $\lambda$.
	\item If $w$ plays \zero and at least one of $u_1,\ldots, u_k$ plays \one, then the payoff for $w$ is 0.
	\item If $w$ plays \one and every one of $u_1,\ldots, u_k$ plays \zero, then the payoff for $w$ is 0.
	\item If $w$ plays \one and at least one of $u_1,\ldots, u_k$ plays \one, then the payoff for $w$ is $1$.
\end{itemize}
We are now ready to prove that this construction correctly simulates a \PURE gate. We consider the different cases that arise depending on the value encoded by $u$.
\begin{itemize}
	\item[{\bf --}] $s_u(\zero)=1$, i.e., $u$ encodes 0. From Lemma~\ref{lem:graphical-wsne-aux} we know that $s_{u_i}(\zero)=1$ for every $i \in [d]$. Then, we have the following for players $v$ and $w$.
	\begin{itemize}
		\item[-] Player $v$ gets payoff 1 from action \zero and payoff 0 from action \one. Hence, in an \eps-WSNE we get that $s_v(\zero)=1$, and thus $v$ encodes 0.
		\item[-] Player $w$ gets payoff $\lambda$ from action \zero and payoff 0 from action \one. Hence, since $\eps < \lambda$, in an \eps-WSNE we get that $s_w(\zero)=1$, and thus $w$ encodes 0.
	\end{itemize}
	\item[{\bf --}] $s_u(\one)=1$, i.e., $u$ encodes 1. From Lemma~\ref{lem:graphical-wsne-aux} we know that $s_{u_i}(\one)=1$ for every $i \in [d]$. Then, we have the following for players $v$ and $w$.
	\begin{itemize}
		\item[-] Player $v$ gets payoff 0 from action \zero and payoff $\lambda$ from action \one. Hence, since $\eps < \lambda$, in an \eps-WSNE we get that $s_v(\one)=1$, and thus $v$ encodes 1.
		\item[-] Player $w$ gets payoff 0 from action \zero and payoff 1 from action \one. Hence, in an \eps-WSNE we get that $s_w(\one)=1$, and thus $w$ encodes 1.
	\end{itemize}
	\item[{\bf --}] $s_u(\one) \in (0,1)$, i.e., $u$ encodes \garbo. Then each one of the auxiliary players $u_1, \ldots, u_d$ can play a different strategy. For each $i \in [d]$, denote $s_{u_i}(\one) = p_i$, i.e., $p_i \in [0,1]$ is the probability player $u_i$ assigns on action \one. Let $P:= \prod_{i \in [d]} p_i$ and $Q := \prod_{i \in [d]} (1-p_i)$.
	Then, we have the following two cases.
	\begin{itemize}
		\item[-] $P \leq 2^{-d}$. Then we focus on player $v$: action \zero yields expected payoff $1-P \geq 1-2^{-d}$, while action \one yields expected payoff $P\cdot \lambda \leq 2^{-d} \cdot \lambda$. Then, since $\eps < \lambda$, we get that in an \eps-WSNE it must hold that $s_v(\zero) = 1$, i.e., $v$ encodes 0.
		\item[-] $P > 2^{-d}$. Then, it holds that $Q < 2^{-d}$; this is because $P\cdot Q = \prod_{i \in [d]} p_i \cdot (1-p_i) \leq (1/4)^d$.
		In this case we focus on player $w$: action \zero yields payoff $\lambda \cdot Q < \lambda \cdot 2^{-d}$, while action \one yields payoff $1-Q > 1-2^{-d}$. Then, again since $\eps < \lambda$, in any \eps-WSNE it must hold that $s_w(\one) = 1$, i.e., $w$ encodes 1.
	\end{itemize}
\end{itemize}

From the arguments given above, we have that in any \wn{\eps} of the
graphical game, with $\eps < 1 - \frac{2}{2^d + 1}$, the players correctly encode a solution to the \pcircuit instance.
\begin{theorem}
	\label{thm:wsne-graph}
	Computing an \wn{\eps} in two-action graphical games with maximum in-degree $d \geq 2$ is \ppad/-complete for any $\eps < 1 - \frac{2}{2^d + 1}$.
\end{theorem}

We can see that the constructed game is not win-lose since there is a payoff $\lambda \notin \{0,1\}$ in the gadget that simulates \PURE gates. However, if we set $\lambda = 1$, and use verbatim the analysis from above, we will get \ppad/-hardness for \eps-WSNE with  $\eps < 1 - \frac{1}{2^{d-1}}$.

\begin{theorem}
	\label{thm:wsne-graph-win-lose}
	Computing an \wn{\eps} in two-action  win-lose graphical games with maximum in-degree $d \geq 2$ is \ppad/-complete for any $\eps < 1 - \frac{1}{2^{d-1}}$.
\end{theorem}

Since for every constant $\eps < 1$ we can find a constant $d$ such that Theorem~\ref{thm:wsne-graph-win-lose} holds, we get the following corollary.

\begin{corollary}
	\label{cor:wsne-PPAD-any-eps}
	For any constant $\eps < 1$, computing an \eps-WSNE in two-action win-lose graphical games is \ppad/-complete.
\end{corollary}

\section{Approximate Nash Equilibria}

\paragraph{\bf A straightforward upper bound.}

We first show that a $0.5$-NE can easily be 
found in any two-action graphical game. 

\begin{theorem}
	There is a polynomial-time algorithm that finds a $0.5$-NE in a two-action
	graphical game.
\end{theorem}
\begin{proof}
	Let $\vbs$ be the strategy profile in
	which all players mix uniformly over their two actions.
	Then, for each player $i$ we have
	\begin{align*}
		u_i(\vbs) &\ge 0.5 \cdot \br_i(\vbsi) \\
		&\ge \br_i(\vbsi) - 0.5,
	\end{align*}
	where the final inequality used the fact that $\br_i(\vbsi) \in [0, 1]$. 
	Thus, $\vbs$ is a $0.5$-NE.
\end{proof}

\paragraph{\bf The lower bound.}

We will show that computing a $(0.5 - \eps)$-NE of a graphical game is \ppad/-hard
for any constant $\eps > 0$ by a reduction from \pcircuit. 

Given a \pcircuit instance, 
we build a two-action graphical game, where the two actions will be named 
$\zero$ and $\one$.
Each node $v$ of the \pcircuit instance will be represented by a set of $k$
players in the game, named $v_1, v_2, \dots, v_k$, where we fix $k$ to be an odd 
number satisfying $k \ge \ln (12/\eps) \cdot 18/\eps^2$. 
Therefore, since $\eps$ is constant, $k$ is also constant.

The strategies of these players will
encode a solution to the \pcircuit problem in the following way. Given a
strategy profile \vbs, we define the assignment $\valonly$ such that
\begin{itemize}
	\item If $s_{v_i}(\zero) \ge 0.5 + \eps/3$ for all $i$, then $\val{v} = 0$.
	\item If $s_{v_i}(\one) \ge 0.5 + \eps/3$ for all $i$, then $\val{v} = 1$.
	\item In all other cases, $\val{v} = \garbo$.
\end{itemize}

We now give implementations for \NOT, \AND, and \PURE gates. We note that, in
all three cases, the payoff received by player $v_i$ is only affected by the
actions chosen by the players representing the inputs to the (unique) gate $g$
that outputs to $v$. Thus, we can argue about the equilibrium condition at $v_i$
by only considering the players involved in gate $g$, and we can ignore all
other gates while doing this.

\paragraph{\bf \NOT gates.}

For a gate $g = (\NOT, u, v)$, we use the following construction, which
specifies the games that will be played between the set of players that represent
$u$ and the set of players that represent $v$. 

Each player $v_i$ has incoming edges from all players $u_1, u_2, \dots, u_k$, and $v_i$'s payoff tensor is set as follows. 
\begin{itemize}
	\item If strictly more than\footnote{Since $k$ is odd, it is not possible for
		exactly $k/2$ players to play $\zero$.} $k/2$ of the players $u_1$ through $u_k$ play
	$\zero$, then $v_i$ receives payoff $0$ for strategy $\zero$ and payoff 1 for
	strategy $\one$.
	\item If strictly less than $k/2$ of the players $u_1$ through $u_k$ play
	$\zero$, then $v_i$ receives payoff $1$ for strategy $\zero$ and payoff 0 for
	strategy $\one$.
\end{itemize}

We now show the correctness of this construction. We start with a technical
lemma that we will use repeatedly throughout the construction.

\begin{lemma}
	\label{lem:payoff2prob}
	Suppose that player $p$ has two actions $a$ and~$b$. 
	In any $(0.5 - \eps)$-NE,
	if the payoff of
	action $a$ is at most $\eps/3$, and the payoff of action $b$ is at least $1
	- \eps/3$, then player~$p$ must play action $b$ with probability at least $0.5 + \eps/3$. 
\end{lemma}
\begin{proof}
	Since action $b$ has payoff at least $1 - \eps/3$, we have that the best
	response payoff to $p$ is also at least $1 - \eps/3$. Hence, in any strategy
	profile \vbs that is an $(0.5 -
	\eps)$-NE, we have
	\begin{align*}
		u_{p}(\vbs) &\ge \br_{p}(\vbs) - 0.5 + \eps \\
		&\ge 1 - \eps/3 - 0.5 + \eps \\
		&= 0.5 + 2\eps/3.
	\end{align*}
	Since the payoff of $a$ is bounded by $\eps/3$, and the payoff of $b$ is bounded
	by $1$ (since all payoffs are in the range $[0, 1]$), we obtain
	\begin{align*}
		u_{p}(\vbs) &\le s_{p}(a) \cdot \eps/3 + s_{p}(b) \cdot 1 \\
		&=  (1 - s_{p}(b)) \cdot \eps/3 + s_{p}(b) \cdot 1 \\
		&= s_{p}(b) (1 - \eps/3) + \eps/3 \\
		&\le s_{p}(b) + \eps/3.
	\end{align*}
	Joining the two previous inequalities gives 
	$s_{p}(b) + \eps/3 \ge 0.5 + 2\eps/3$,
	and therefore 
	$s_{p}(b) \ge 0.5 + \eps/3$. 
\end{proof}

We can now prove that the \NOT gadget operates correctly.

\begin{lemma}
	In every $(0.5-\eps)$-NE, the following properties hold.
	\begin{itemize}
		\item If the players representing $u$ encode $0$, then the players
		representing $v$ encode $1$.
		\item If the players representing $u$ encode $1$, then the players
		representing $v$ encode $0$.
	\end{itemize}
\end{lemma}
\begin{proof}
	Let $\vbs$ be a $(0.5-\eps)$-NE. 
	We begin with the first claim. Since the players representing $u$ encode a $0$,
	we have that $s_{u_j}(\zero) \ge 0.5 + \eps/3$ for all~$j \in [k]$. 
	
	We start by proving an upper bound on the payoff of action \zero for player
	$v_i$. The payoff of this action increases as the players $u_j$ place
	less probability on action \zero, so we can assume 
	$s_{u_j}(\zero) = 0.5 + \eps/3$, since this minimizes the payoff of \zero to
	$v_i$. 
	
	Under this assumption, the number $N$ of players $u_j$ that play action \zero is distributed
	binomially according to $N \sim B(k, 0.5 + \eps/3)$. 
	Applying the standard Hoeffding bound~\cite{hoeffding1994probability} for the binomial distribution, and using
	the fact that $k \ge \ln (6/\eps) \cdot 9/2\eps^2$ yields the following 
	\begin{align*}
		\Pr(N \le k/2) &\le 2 \cdot \exp\left(-2k \left( 0.5 + \eps/3 - \frac{k/2}{k} \right)^2
		\right) \\
		& = 2 \cdot \exp\left(-2k \cdot \eps^2/9 \right) \\
		& \le \exp\left(-\ln(3/\eps) \right) \\
		& = \eps/3.
	\end{align*}
	Hence the payoff of action $\zero$ to player $v_i$ is at most
	$\eps/3$, and therefore the payoff of action $\one$ to $v_i$ is at least $1 -
	\eps/3$.

	So we can apply Lemma~\ref{lem:payoff2prob} to argue that $s_{v_i}(\one) \ge 0.5
	+ \eps/3$. Since this holds for all~$i$, we have that the players representing
	$v$ encode the value $1$ in the \pcircuit instance, as required.
	
	The second case can be proved in an entirely symmetric manner.
\end{proof}

\paragraph{\bf \AND gates.} 

For a gate $g = (\AND, u, v, w)$ we use the following construction. Each player
$w_i$ has in-degree $2k$ and has incoming edges from all of the players $u_1, u_2, \dots, u_k$, and all of the players $v_1, v_2, \dots, v_k$. 
The payoff tensor of $w_i$ is as follows.

\begin{itemize}
	\item If strictly more than $k/2$ of the players $u_1$ through $u_k$ play \one,
	and strictly more than $k/2$ of the players $v_1$ through $v_k$ play \one, 
	then $w_i$ receives payoff $0$ from action $\zero$ and payoff 1 from
	action $\one$.
	
	\item If this is not the case, then 
	$w_i$ receives payoff $1$ from action $\zero$ and payoff 0 from action $\one$. 
\end{itemize}

Correctness of this construction is shown in the following pair of lemmas.

\begin{lemma}
	In every $(0.5-\eps)$-NE of the game, if the players representing $u$ encode value
	$1$, and the players representing $v$ encode value $1$, then the players
	representing $w$ will encode value $1$.
\end{lemma}
\begin{proof}
	Let $\vbs$ be a $(0.5-\eps)$-NE. From the assumptions about $u$ and $v$, we have that
	$s_{u_j}(\one) \ge 0.5 + \eps/3$ for all $j$, and $s_{v_j}(\one) \ge 0.5 +
	\eps/3$ for all $j \in [k]$.
	
	We start by proving an upper bound on the payoff of $\zero$ to $w_i$. Since this payoff
	decreases as the players $u_j$ and $v_j$ place more probability on $\one$, we
	can assume that 
	$s_{u_j}(\one) = 0.5
	+ \eps/3$ for all $j$, and $s_{v_j}(\one) = 0.5 + \eps/3$ for all $j$, since
	this maximizes the payoff of $\zero$ to $w_i$. 
	
	The number $N$ of players $u_j$ that play \one is distributed binomially according
	to $N \sim B(k, 0.5 + \eps/3)$. Similarly, the number of players $v_j$ that
	play \one, are distributed according to the same distribution, that is why we focus only in the former. Using the standard Hoeffding bound for the
	binomial distribution, along with the fact that 
	$k \ge \ln (12/\eps) \cdot 9/2\eps^2$, we get
	\begin{align*}
		\Pr(N \le k/2) &\le 2 \cdot \exp\left(-2k \left( 0.5 + \eps/3 - \frac{k/2}{k} \right)^2
		\right) \\
		& = 2 \cdot \exp\left(-2k  \cdot \eps^2/9 \right) \\
		& \le \exp\left(-\ln(6/\eps) \right) \\
		& = \eps/6.
	\end{align*}
	We can then use the union bound to prove that the probability that 
	strictly less than $k/2$ of the players $u_1$ through $u_k$ play \one,
	or strictly less than $k/2$ of the players $v_1$ through $v_k$ play \one is at
	most $\eps/3$.
	
	Hence, the payoff of $\zero$ to player $w_i$ is at most $\eps/3$, and so the
	payoff of $\one$ to player $w_i$ is at least $1 - \eps/3$. Thus we can apply
	Lemma~\ref{lem:payoff2prob} to argue that $w_i$ must play $\one$ with
	probability at least $0.5 + \eps/3$. Since this holds for all $i$, we have that
	$w_1, w_2, \ldots w_k$ correctly encode value $1$.
\end{proof}

\begin{lemma}
	In every $(0.5-\eps)$-NE of the game, if the players representing $u$ encode value
	$0$, or the players representing $v$ encode value $0$, then the players
	representing $w$ will encode value $0$.
\end{lemma}

\begin{proof}
	We will provide a proof for the case where the players representing $u$ encode
	value $0$, since the other case is entirely symmetric. 
	
	Let $\vbs$ be an $(0.5-\eps)$-NE. By assumption we have that $s_{u_j}(\zero) \ge 0.5 +
	\eps/3$ for all $j$.
	We start by proving an upper bound on the payoff of $\one$ to $w_i$. Since the
	payoff of this strategy decreases as $u_j$ places more probability on $\zero$,
	we can assume that 
	$s_{u_j}(\zero) = 0.5 + \eps/3$ for all $j$, since this minimizes the payoff of
	$\one$ to $w_i$. 
	
	The number of players $u_j$ that play \one is distributed binomially according
	to $N \sim B(k, 0.5 + \eps/3)$. Using the standard Hoeffding bound for the
	binomial distribution, along with the fact that 
	$k \ge \ln (6/\eps) \cdot 9/2\eps^2$, we get
	\begin{align*}
		\Pr(N \le k/2) & \le 2 \cdot  \exp\left(-2k \left( 0.5 + \eps/3 - \frac{k/2}{k} \right)^2
		\right) \\
		& = 2 \cdot \exp\left(-2k \cdot \eps^2/9 \right) \\
		& \le \exp\left(-\ln(3/\eps) \right) \\
		& = \eps/3.
	\end{align*}
	Hence, the payoff of $\one$ to player $w_i$ is at most $\eps/3$, and so the
	payoff of $\zero$ to player $w_i$ is at least $1 - \eps/3$. Thus, we can apply
	Lemma~\ref{lem:payoff2prob} to argue that $w_i$ must play $\zero$ with
	probability at least $0.5 + \eps/3$. Since this holds for all $i$, we have that
	$w_1, w_2, \ldots w_k$ correctly encode value $0$.
\end{proof}

\paragraph{\bf \PURE gates.}

For a gate $g = (\PURE, u, v, w)$, we use the following construction. Each
player $v_i$ will have incoming edges from all players $u_1, u_2, \dots, u_k$, with 
their payoff tensor set as follows.
\begin{itemize}
	\item If strictly more than $(0.5 - \eps/6) \cdot k$ of the players $u_1$
	through $u_k$ play $\one$, then $v_i$ receives payoff $0$ from
	action $\zero$ and payoff $1$ from action $\one$. 
	
	\item If this is not the case, then $v_i$ receives payoff $1$ from action
	$\zero$ and payoff $0$ from action $\one$.
	
\end{itemize}
Each player $w_i$ 
will have incoming edges from all players $u_1, u_2, \dots, u_k$, with their payoff tensor set as follows.
\begin{itemize}
	\item If strictly more than $(0.5 + \eps/6) \cdot k$ of the players $u_1$
	through $u_k$ play $\one$, then $w_i$ receives payoff $0$ from 
	action $\zero$ and payoff $1$ from action $\one$. 
	
	\item If this is not the case, then $w_i$ receives payoff $1$ from action
	$\zero$ and payoff $0$ from action $\one$.
	
\end{itemize}
The following lemma shows that $v$ is correctly simulated.

\begin{lemma}
	\label{lem:nepurev}
	Let $\vbs$ be a $(0.5 - \eps)$-NE, and let $E[X]$ denote the expected number of
	players $u_i$ that play strategy $\one$ under $\vbs$. 
	\begin{itemize}
		\item If $E[X] \le (0.5 - \eps/3) \cdot k$, then the players representing $v$
		will encode value $0$. 
		
		\item If $E[X] \ge 0.5 \cdot k$, then the players representing $v$ will encode
		value $1$. 
	\end{itemize}
\end{lemma}
\begin{proof}
	We will prove only the first case, since the second case can be proved in an
	entirely symmetric manner. 
	Using the fact that $E[X] \le (0.5 - \eps/3) \cdot k$, then applying Hoeffding's
	inequality, and using the fact that $k \ge \ln (3/\eps) \cdot 18/\eps^2$ we get
	\begin{align*}
		\Pr(X \ge (0.5 - \eps/6) \cdot k) 
		&\le \Pr(X - E[X] \ge \eps/6 \cdot k) \\
		& \le \exp\left( \frac{-2 (k \cdot \eps/6)^2}{k} \right) \\
		& = \exp \left( -k \cdot \eps^2/18 \right) \\
		& \le \eps/3.
	\end{align*}
	Therefore, the payoff of strategy $\one$ to $v_i$ is at most $\eps/3$, and so
	the payoff of strategy $\zero$ to $v_i$ is at least $1 - \eps/3$. Thus we can
	apply Lemma~\ref{lem:payoff2prob} to argue that 
	$s_{v_i}(\zero) \ge 0.5 + \eps/3$. Since this applies for all $i$, we have that
	the players representing $v$ encode value $0$, as required.
\end{proof}

\noindent
The next lemma shows that $w$ is also correctly simulated; its proof is omitted since it is entirely symmetric to the proof of Lemma~\ref{lem:nepurev}.

\begin{lemma}
	Let $\vbs$ be a $(0.5 - \eps)$-NE, and let $E[X]$ denote the expected number of
	players $u_i$ that play strategy $\one$ under $\vbs$. 
	\begin{itemize}
		\item If $E[X] \le 0.5 \cdot k$, then the players representing $w$
		will encode value $0$. 
		
		\item If $E[X] \ge (0.5 + \eps/3) \cdot k$, then the players representing $w$ will encode
		value $1$. 
	\end{itemize}
\end{lemma}

\noindent
Combining the two previous lemmas, we can see that the construction correctly
simulates a \PURE gate.

\begin{itemize}
	\item If the players representing $u$ encode value $0$, then $E[X] \le 0.5 -
	\eps/3$, and so both the players representing $v$ and those representing $w$ encode value $0$. 
	\item If the players representing $u$ encode value $1$, then $E[X] \ge 0.5 +
	\eps/3$, and so both the players representing $v$ and those representing $w$ encode value $1$. 
	\item In all other cases we can verify that either the players representing $v$ or the players representing $w$ encode
	a $0$ or a $1$. Specifically, if $E[X] \le 0.5 \cdot k$ then the players representing $w$ encode value $0$,
	while if $E[X] \ge 0.5 \cdot k$, then the players representing $v$ encode value $1$.
\end{itemize}

\paragraph{\bf The hardness result.}

From the arguments given above, we have that in an $(0.5 - \eps)$-NE of the
graphical game, the players correctly encode a solution to the \pcircuit
instance.
Note also that, since Theorem~\ref{thm:pancircuit} gives hardness for
\pcircuit even when the total degree of each node is 3, the graphical game that we
have built has total degree at most $3k$. Thus, the game can be built in polynomial time.

\begin{theorem}
	It is \ppad/-hard to find a $(0.5 - \eps)$-NE in a 
	two-action graphical game for any constant $\eps > 0$.
\end{theorem}

In fact, the payoff entries in all gadgets are 0 or 1. Thus, our \ppad/-hardness result holds for win-lose games too.

\begin{corollary}
	For any constant $\eps > 0$, it is \ppad/-hard to find a $(0.5 - \eps)$-NE in a 
	two-action win-lose graphical game.
\end{corollary}

\section{Conclusions}
We have resolved the computational complexity of finding approximate Nash equilibria in two-action graphical games, by providing complete characterizations for both $\eps$-NE and $\eps$-WSNE.
Our results show that finding approximate Nash equilibria in graphical games is much harder when compared to the special case of (constant-degree) polymatrix games: for two-action polymatrix games the tractability threshold for \eps-WSNE is $1/3$ \cite{DFHM22}.

Below we identify two research questions that deserve more research. 
\begin{itemize}
	\item What is the intractability threshold for \eps-NE in graphical games with more than two actions? We have shown that $0.5$ is the threshold for two-action games, and 
	we conjecture that $0.5$ is the correct answer for the multi-action case as well. Hence, we view the main open problem as finding a polynomial-time algorithm that finds a $0.5$-NE in any graphical game. 
	We note that such an algorithm is already known for the special case of polymatrix games~\cite{DFSS17}.
	\item What is the intractability threshold for \eps-NE and \eps-WSNE in polymatrix games? 
	For \eps-NE our understanding is far from complete, even in the two-action case, since there is a substantial gap
	between the $1/48$ lower bound and the $1/3$ upper bound (where the latter actually comes from the  tractability of $1/3$-WSNE) given in \cite{DFHM22}.
	For \eps-WSNE, although the problem is completely resolved for the two-action case, the gap in multi-action polymatrix games is still large, and it seems that improving either the known lower bound of $1/3$, or the trivial upper bound of $1$, would require significantly new ideas.
\end{itemize}

\subsubsection*{Acknowledgements}
The second author was supported by EPSRC grant EP/W014750/1 ``New Techniques for Resolving Boundary Problems in Total Search''.

\bibliographystyle{alphaurl}
\bibliography{refs}

\end{document}